%% file: main.tex
\documentclass[letterpaper, 10 pt, conference]{ieeeconf}   
\IEEEoverridecommandlockouts                               
\usepackage{amssymb,amsmath,color,subfigure}
\usepackage{graphicx}
\usepackage{xspace,stackrel,hyperref}
\usepackage{wrapfig}
\usepackage{mathtools}
\usepackage{tikz,pgfplots}
\usepackage{etoolbox}
\usetikzlibrary{external} 
\usepackage{autonum}
\usepackage{dsfont}
\usepackage{etoolbox}

\usepackage{multirow}
\usepackage{cite}

\DeclareMathOperator*{\argmin}{arg\,min}

\definecolor{mycolor4}{RGB}{230,97,1}
\definecolor{mycolor2}{RGB}{178,171,210}
\definecolor{mycolor3}{RGB}{253,184,99}
\definecolor{mycolor1}{RGB}{94,60,153}

\makeatletter 
\pretocmd\@bibitem{\color{black}\csname keycolor#1\endcsname}{}{\fail}
\newcommand\citecolor[1]{\@namedef{keycolor#1}{\color{blue}}}
\makeatother

\DeclareMathAlphabet{\mathcal}{OMS}{cmsy}{m}{n}

\def\beq{\begin{equation}}
\def\eeq{\end{equation}}

\newcommand{\mc}{\mathcal}

\newcommand{\Z}{\mathbb{Z}}

\newcommand{\R}{\mathds{R}}

\newcommand{\defineas}{\coloneqq}

\newcommand{\norm}[1]{\left\lVert#1\right\rVert}

\definecolor{mycolor1}{RGB}{230,97,1}
\definecolor{mycolor2}{RGB}{178,171,210}
\definecolor{mycolor3}{RGB}{253,184,99}
\definecolor{mycolor4}{RGB}{94,60,153}
\definecolor{mycolor5}{rgb}{0,0,0}

\tikzset{
  pics/car/.style args={#1}{
     code={
     \begin{scope}[scale=0.15]
      \shade[top color=#1, bottom color=white, shading angle={135}]
        [draw=black,fill=red!20,rounded corners=0.2ex] (1.5,.5) -- ++(0,1) -- ++(1,0.3) --  ++(3,0) -- ++(1,0) -- ++(0,-1.3) -- (1.5,.5) -- cycle;
    \draw[ rounded corners=0.5ex,fill=black!20!blue!20!white]  (2.5,1.8) -- ++(1,0.7) -- ++(1.6,0) -- ++(0.6,-0.7) -- (2.5,1.8);
    \draw[thick]  (4.2,1.8) -- (4.2,2.5);
    \draw[draw=black,fill=gray!50,thick] (2.75,.5) circle (.5);
    \draw[draw=black,fill=gray!50,thick] (5.5,.5) circle (.5);
    \end{scope}
     }
  }
}

\newtheorem{theorem}{Theorem}
\newtheorem{proposition}{Proposition}

\newtheorem{definition}{Definition}

\newtheorem{remark}{Remark}

\newtheorem{problem}{Problem}

\newtoggle{full_version}
\toggletrue{full_version}

\title{\LARGE \bf
A Pricing Mechanism for Balancing the Charging of Ride-Hailing Electric Vehicle Fleets 
}

\author{Marko Maljkovic, Gustav Nilsson, and Nikolas Geroliminis% <-this % stops a space
\thanks{M.~Maljkovic, G.~Nilsson, and N.~Geroliminis are with the School of Architecture, Civil and Environmental Engineering, École Polytechnique Fédérale de Lausanne (EPFL), 1015 Lausanne, Switzerland. {\tt\small \{marko.maljkovic, gustav.nilsson, nikolas.geroliminis\}@epfl.ch}.}%
\thanks{This work was supported by the Swiss National Science Foundation under NCCR Automation, grant agreement 51NF40\_180545.}
\iftoggle{full_version}{}{\thanks{An extended version containing all the proofs is avilable at \url{https://arxiv.org/abs/2203.0610X}}}%
}

\begin{document}

\maketitle
\thispagestyle{empty}
\pagestyle{empty}

\begin{abstract}
Both ride-hailing services and electric vehicles are becoming increasingly popular and it is likely that charging management of the ride-hailing vehicles will be a significant part of the ride-hailing company's operation in the near future. Motivated by this, we propose a game theoretic model for charging management, where we assume that it is the fleet-operator that wants to minimize its operational cost, which among others include the price of charging. To avoid overcrowded charging stations, a central authority will design pricing policies to incentivize the vehicles to spread out among the charging stations, in a setting where several ride-hailing companies compete about the resources. We show that it is possible to construct pricing policies that make the Nash-equilibrium between the companies follow the central authority's target value when the desired load is feasible. Moreover, we provide a decentralized algorithm for computation of the equilibrium and conclude the paper with a numerical example illustrating the results.
\end{abstract}

\section{Introduction}

Ride-hailing services have become more and more popular over the past years and are nowadays an essential part of transportation services in many cities. Also, electric vehicles (EVs) are becoming more common and are soon likely to become a significant part of the fleets of vehicles that ride-hailing companies manage. Since ride-hailing companies already offer access to cleaning and service stations to their drivers, it is not unlikely that in the future, they will offer discounted charging as well. By doing so, the companies can gain control of both the coverage by sending vehicles to charge in areas where there is demand and the availability, e.g., by incentivizing the drivers to charge up their vehicles before demand peaks. Moreover, given the asymmetric distribution of origins and destinations, pricing incentives can contribute also in rebalancing vehicles in regions of higher demand. In the case of autonomous fleets, the ride-hailing company would have total control of the vehicles and could also fully control the charging.

Inspired by this vision, this paper presents a pricing mechanism to load balance the ride-hailing vehicles among different charging stations. We study a scenario where a central body, e.g., the government of the city or the power providing company, defines the set points describing how the vehicles should spread out among the charging facilities in an attempt to either help fight the congestion in the city or to balance the demand on the power grid. The central body is incentivizing the ride-hailing companies to follow the desired set points through pricing while each company is trying to optimize its operational cost by directing its vehicles to different charging stations. A schematic representation of the problem is shown in Figure~\ref{fig:problem}. Due to every company's interest in minimizing its queuing time at the stations, there is an inherent competition among them establishing fertile ground for game theoretic analysis.

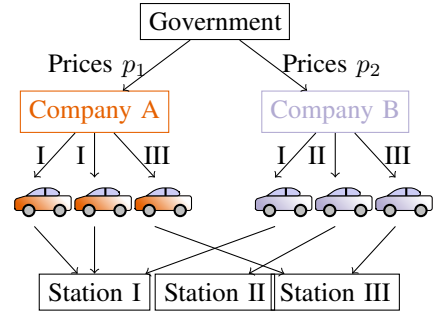
\begin{figure}
    \centering
    \begin{tikzpicture}[scale=0.8]
        \node[draw] (gov) at (0,0) {Government};
        \node[draw,mycolor1] (comA) at (-2,-1.5) {Company A};
        \node[draw, mycolor2] (comB) at (2,-1.5) {Company B};

       \node (veh1) at (-3,-3) {\begin{tikzpicture}\draw (0,0) pic{car=mycolor1}; \end{tikzpicture}};
       \node (veh2) at (-2,-3) {\begin{tikzpicture}
 \draw (0,0) pic{car=mycolor1}; \end{tikzpicture}};
       \node (veh3) at (-1,-3) {\begin{tikzpicture}
 \draw (0,0) pic{car=mycolor1}; \end{tikzpicture}};

       \node (veh4) at (1,-3) {\begin{tikzpicture}
 \draw (0,0) pic{car=mycolor2}; \end{tikzpicture}};
       \node (veh5) at (2,-3) {\begin{tikzpicture}
 \draw (0,0) pic{car=mycolor2}; \end{tikzpicture}};
       \node (veh6) at (3,-3) {\begin{tikzpicture}
 \draw (0,0) pic{car=mycolor2}; \end{tikzpicture}};
       
       \node[draw] (sta1) at (-2,-4.5) {Station I};
       \node[draw] (sta2) at (0,-4.5) {Station II};
       \node[draw] (sta3) at (2,-4.5) {Station III};
       %\node (prices) at (0, -.75) {Prices $p$};
        \draw[->] (gov) -- node[left] {Prices $p_1$} (comA)  ;
        \draw[->] (gov) -- node[right] {Prices $p_2$} (comB);
        
        \draw[->] (comA) -- node[left] {I} (veh1.north);
        \draw[->] (comA) -- node[left] {I} (veh2.north);
        \draw[->] (comA) -- node[right] {III} (veh3.north);
        
        \draw[->] (comB) -- node[left] {I} (veh4.north);
        \draw[->] (comB) -- node[left] {II} (veh5.north);
        \draw[->] (comB) -- node[right] {III} (veh6.north);

        \draw[->] (veh1.south) -- (sta1);
        \draw[->] (veh2.south) -- (sta1);
        \draw[->] (veh3.south) -- (sta3);
        
        \draw[->] (veh4.south) -- (sta1);
        \draw[->] (veh5.south) -- (sta2);
        \draw[->] (veh6.south) -- (sta3);

    \end{tikzpicture}

    \caption{Schematic sketch of the problem setting. The central body, e.g., the government or the power company, wants to balance the vehicle load on different charging stations through pricing policies. With the provided pricing policies, each ride-hailing company wants to minimize its own operational cost by steering its vehicles to different charging stations.}
    \label{fig:problem}
\end{figure}

Research has shown that the frameworks of congestion, mean-field, Stackelberg, and inverse Stackelberg games are powerful tools for solving problems within the realm of transportation and mobility systems. In \cite{article1, article2}, through congestion game based routing, tolling mechanisms have been designed for congestion control of urban networks, whereas in \cite{article3}, charging station allocation for a population of EVs has been performed. The structure of the mean-field games, where the utility of each player depends on the aggregation of other players' decisions, offers a suitable setting for charging control of a population of EVs as presented in \cite{Paccagnan2016a, Paccagnan2016b, Paccagnan2019, EVsCharging}. Our work is similar to \cite{Paccagnan2016a, Paccagnan2016b, Paccagnan2019, EVsCharging} in a sense that the underlying structure of our problem can also be described by an aggregative game. However, we also go along the line of research that focuses on Stackelberg and inverse Stackelberg games to design pricing and tolling mechanisms primarily for revenue maximization. In \cite{article4, article5}, the charging stations act as revenue maximizing leaders in a Stackelberg game, whereas individual EVs act as charging cost minimizing followers. The setup in \cite{article4, article5} assumes fixed optimal prices of charging. In this paper, we propose using a pricing mechanism based on the decision of the ride-hailing companies which allows us to directly influence the placement of the Nash equilibrium. This makes our setup more similar to the ones presented in \cite{article6, article7} where inverse Stackelberg game has been used to solve hierarchical control and bi-level optimal toll design problems. The inverse Stackelberg pricing schemes are different from the Stackelberg ones in a sense that the prices are not a priori set to a certain value by the leading player, i.e., the central body, but are rather announced as a function of the followers' decisions. This means that the companies do not know what the charging prices will be before they make a decision on how to direct their vehicles but rather how their joint decision will influence the prices of charging.

To the best of our knowledge, no work so far has provided a comprehensive framework for analysing the problem of balancing the charging of EV fleets operated by ride-hailing companies so as to achieve the objective of a higher level authority. Moreover, we do so in a decentralized manner, with little private information exchange between the ride-hailing companies and the government and under the reachability constraints imposed by the state of the individual car fleets.              

The paper is outlined as follows: the rest of this section is devoted to introducing some basic notation. In Section~\ref{sec:model} we introduce the model and state the main formulation. In the following section, Section~\ref{sec:pricing}, we present the pricing mechanism and show that this pricing mechanism achieves a unique Nash-equilibrium between the companies. We also provide an algorithm to compute the Nash-equilibrium. In Section~\ref{sec:example}, we illustrate the proposed solution through a numerical example and conclude the paper  with some ideas for future research in Section~\ref{sec:conclusion}.

\subsection{Notation}
Let $\R$ denote the set of real numbers, and $\R_+$ the set of non-negative reals. Let $\mathbf{0}_{m}$ and $\mathbf{1}_m$ denote the all zero and all one vectors of length $m$ respectively, and $\mathbb{I}_m$ the identity matrix of size $m \times m$. For a finite set $\mc A$, we let $\R_{(+)}^{\mc A}$ denote the set of (non-negative) vectors indexed by the elements of $\mc A$, $\left|\mc A\right|$ the cardinality of $\mc A$ and we let $\mc P_{\mc A}$ be the probability space over the set, i.e., $\mc P_{\mc A} \defineas \{ x \in \R_+^{\mc A} \mid \sum_{i \in \mc A} x_i = 1\}$. For a diagonal matrix $A \in \R^{n \times n}$, we let $A^*$ denote its pseudo-inverse, i.e., 
\begin{equation}
A^*_{ii} \defineas 
\begin{cases} 1/A_{ii} & \textrm{if } A_{ii} \neq 0, \\
0 & \textrm{otherwise,}
\end{cases}\quad 1 \leq i \leq n\,.
\end{equation}

\section{Model}\label{sec:model}
We consider a setting where different ride-hailing companies have access to common charging stations for their electric vehicles. We let $\mc C$ denote the set of companies, and $N_i >0$ the number of vehicles belonging to each company $i \in \mc C$. The vector of the number of vehicles for all companies is denoted by $N \in \R_+^{\mc C}$. We let $\mc M$ represent the set of charging stations, and $M_j > 0$ the number of spots available at each charging station $j \in \mc M$, i.e., the charging station's capacity. The vector of all charging stations capacities is denoted by $M \in \R_+^{\mc M}$ and the cardinality of $\mc M$ as $m=|\mc M|$.

For each company $i \in \mc C$, we let $\mc V_i$ be the set of its vehicles with $\left|\mc V_i\right|=N_i$ and $x^{i} \in \mc P_{\mc M}$ denote the fraction of vehicles that the company wants to send to each charging station, i.e., $x^i_j$ is the fraction of vehicles from company $i \in \mc C$ that will be sent to charging station $j \in \mc M$. Furthermore, let $n^{i}_{j}\in\Z_{+}$ denote the integer number of vehicles, associated with the continuous allocation $x^{i}_{j}$, that the operator of the fleet would send to station $j$. Since not all charging stations are reachable for all vehicles and hence not all choices of $x^{i}$ are feasible, we define for each company the feasibility sets $\mc F_{j}^{i} \defineas \left\{v \in \mc V_i \mid v \text{ can reach station } j \right\}$. 

We say that a continuous allocation vector $x^{i}$ is feasible, if it allows the operator of the company to choose any discrete allocation $n^{i}\in\Z^{\mc M}_{+}$ where individual $n^{i}_{j}$ can be either $\left\lfloor{N_{i}x^{i}_{j}}\right\rfloor$ or $\left\lceil{N_{i}x^{i}_{j}}\right\rceil$ under the constraints that $\sum_{j\in\mc M}n^{i}_{j}=N_{i}$ and that there exists a feasible matching between the vehicles and the charging stations for the chosen $n^{i}$. For each company $i\in\mc C$, we let $\mc K_{i}\subseteq \mc P_{\mc M}$ denote the set of all feasible $x^{i}$. Furthermore, we define $\mc K\defineas \prod_{i\in\mc C} \mc K_{i}$ and $\mc K_{-i}\defineas \prod_{j\in\mc C\setminus i} \mc K_{j}$.

%In the following section, we show how the sets $\mc K_{i}$ can be constructed based on feasibility sets $\mc F_{j}^{i}$ and that they are in fact a polytopic constraint on continuous allocation variables. 

We let $x \defineas \left[x^{i}\right]_{i\in \mc C}\in \mc K$ denote all companies' decision vectors, $x^{-i} \defineas \left[x^{j}\right]_{j\in \mc C \setminus i}\in\mc K_{-i}$ denote the decision vectors of all companies except the company $i$, $\sigma\left(x\right) \defineas \sum_{i \in \mc C} N_{i}x^{i} \in \R_+^{\mc M}$ denote the vector consisting of the total number of vehicles that have chosen each station and $\sigma\left(x^{-i}\right) \defineas \sum_{j \in \mc C \setminus i} N_j x^j \in \R_+^{\mc M}$ denote the vector consisting of the number of vehicles from the other companies that have chosen each station.

To easily distinguish between the agents, we refer to the central authority as the ``government''. It is interested in balancing the vehicles so as to minimize the personal objective of the form
\begin{equation}
    J_G (\sigma\left(x\right)) = \frac{1}{2}\sigma\left(x\right)^{T}A_{G} \sigma\left(x\right)+b_{G}^{T}\sigma\left(x\right) \,,
    \label{eq:JG}
\end{equation}
for some diagonal matrix $A_{G}\succ 0$ and $b_{G}\in \R^{\mc M}$. In this paper, we are particularly interested in balancing the vehicles so that the number of vehicles charging at each station equals $\hat{N}$, i.e., to minimize  
\begin{equation}
    J_G (\sigma\left(x\right)) = \frac{1}{2}\|\sigma(x)-\widehat{N}\|^{2}_{2,A_{G}} \,,
    \label{eq:JG2}
\end{equation}
where $A_G$ gives the government the possibility to penalize deviations from the desired number of vehicles differently at different stations. It should be noted that~\eqref{eq:JG2} is a special case of~\eqref{eq:JG} that can be obtained by letting in $b_{G}=-A_{G}\widehat{N}\in \R^{\mc M}$. 

To steer the companies to the minimum of~\eqref{eq:JG2}, the government will assign an individual pricing policy to each company for each charging station. The policy will be a function of the choice of the company itself but also of the other companies' choices, since the government's interest is to control the total number of vehicles.
%since the waiting time at each charging station due to its limited capacity is dictated by the total number of vehicles that have chosen each of the stations. 
For company $i\in \mc C$, the pricing policy is $p_{i}\left(x^{i},x^{-i}\right) : \mc K_{i}\times \mc K_{-i}\rightarrow \R^{\mc M} $.

After the pricing policies are announced, the government and the companies admit an inverse Stackelberg game in which every company is trying to minimize its own operational cost, under the constraint that all the company's vehicles must be able to reach a charging station. We model the operational cost for each company as a sum of three terms. The first term, denoted as the queuing cost, depends on the choice of the company itself but also on the cumulative choice of all other companies and has the general form 
\begin{equation}
J_{1}^{i}\left(x^{i}, \sigma(x)\right)=\frac{1}{2} (x^{i})^T A_{i}x^{i}+(x^{i})^TB_{i}\sigma\left(x^{-i}\right)+c_{i}^{T}x^{i} \,,
\label{eq:J1igen}
\end{equation}
for some diagonal matrices $A_{i} \in \R^{\mc M \times \mc M}$, $B_{i} \in \R^{\mc M \times \mc M}$ and $c_{i}\in\R^{\mc M}$. In this paper, we model the expected queuing cost as 
$J_{1}^{i}\left(x^{i}, \sigma(x)\right)=N_{i}\left(x^{i}\right)^{T}Q\left(\sigma(x)-M\right),$
which is a special case of \eqref{eq:J1igen} if we set $A_{i}\defineas2N_{i}^{2}Q$, $B_{i}\defineas N_{i}Q$ and $c_{i}\defineas -N_{i}QM$.
Here, $Q\in\R^{\mc M\times \mc M}$ is a positive definite diagonal scaling matrix whose diagonal entries describe how expansive  it is to queue in the regions around charging stations. Generally, more congested areas should experience higher queuing costs and hence higher scaling factors. We model the second term which describes the charging cost as a function of the choice of the company and the pricing policy assigned to it, i.e.,
$J_{2}^{i}\left(x^{i}, p_{i}\left(x^{i},x^{-i}\right)\right)=(x^{i})^T D_{i} p_{i}\left(x^{i},x^{-i}\right)$, 
for some diagonal $D_{i}\succeq 0$. The diagonal entry $\left(D_{i}\right)_{kk}$ can be interpreted as the part of the total charging demand to be served at the charging station $k$. The third term we denote as the negative expected revenue and model it as a function of only the company's choices, i.e., $J_{3}^{i}\left(x^{i}\right)=f_{i}^{T} x^{i}$. Here, we interpret the negative expected revenue as the difference between the cost of fleet being idle while traveling to the charging stations and the expected profit in the regions around charging stations after the charging has been completed. The information about the negative expected revenue per vehicle is encoded in $f_{i}$. Hence, the company cost can be in general expressed as
\begin{equation}
    J^{i}\left(x^{i},x^{-i}\right) = J_{1}^{i}\left(x^{i}, \sigma(x)\right)+J_{2}^{i}\left(x^{i}, p_{i}\left(x^{i},x^{-i}\right)\right)+J_{3}^{i}\left(x^{i}\right) \,
    \label{eq:companyloss}
\end{equation}
and each company $i \in \mc C$ would like to allocate its vehicles according to
\begin{equation}
    x^{i*} \in \argmin_{x^{i} \in \mc K_i} J^{i}\left(x^{i},x^{-i}\right) \,. 
    \label{eq:companymin}
\end{equation}

We say that the government and the companies admit a system optimum if there exists $x^{*}$ that minimizes \eqref{eq:JG} and satisfies \eqref{eq:companymin} for all $i\in\mc C$. We will show in the following section that if we can reduce the decision space of the companies to convex subsets $\overline{\mc K}_{i}\subseteq \mc K_{i}$, under the proposed pricing strategies there will be a unique system optimum.

To summarize, we consider the problem of designing prices, such that each company will steer its fleet of vehicles towards predefined target values of vehicle accumulations around different charging stations. A schematic sketch of the problem is shown in Figure~\ref{fig:problem} and the problem is formally stated below.
\begin{problem} \label{problem:theproblem}
Design pricing policies $p_{i}\left(x^{i},x^{-i}\right)$ and the constraint sets $ \overline{\mc K}_{i}\subseteq \mc K_{i}$ such that there is a unique Nash equilibrium of the game $G$ defined as
\begin{equation}
    G\defineas \left\{\min_{x^{i}\in \overline{\mc K}_{i}}J^{i}\left(x^{i},x^{-i}\right),\forall i\in \mc C\right\} \,,
    \label{eq:gameform}
\end{equation}
with $J^{i}$ defined as in~\eqref{eq:companyloss}. Moreover, the Nash equilibrium should be such that it also minimizes the government cost $J_G(x)$
in~\eqref{eq:JG} and the design of the constraint sets $\overline{\mc K}_{i}$ such that existence of a feasible discrete allocation scheme for each company %with minimal round off error
is guaranteed.  
\end{problem}

\section{Pricing Mechanism}\label{sec:pricing}
We begin this section by showing how the sets $\overline{\mc K}_{i}$ can be constructed. With the existence of those sets, we then proceed to introduce a pricing policy that achieves a unique Nash equilibrium for allocating the vehicles of all companies. Moreover, we show that this Nash equilibrium also minimizes the government's cost function which makes it a unique system optimum. In the last part of the section, we propose an algorithm for computing the Nash equilibrium.

%Until now, we have addressed the problem of vehicle allocation in continuous domain. 
In the following proposition, we show how to analytically construct convex sets $\overline{\mc K}_{i} \subseteq \mc K_{i}$ based on feasibility sets $\mc F^{i}_{j}$, that guarantee feasibility of $x^{i}$ defined as in Section \ref{sec:model}.

\begin{proposition}\label{prop:Kbar}
For each company $i \in \mc C$, define the set $\overline{\mc K}_{i}\subseteq\mc K_{i}\subseteq \mc P_{\mc M}$ such that $x^{i}\in\overline{\mc K}_{i}$ if for all proper subsets $S$ of $\mathcal{M}$, it holds that
\begin{equation}
 N_{i} \sum_{j \in S} x^{i}_{j} \leq \max \left\{0,\left|\bigcup_{j \in S} \mc F_{j}^{i}\right|-|S|\right\}  \,.
 \label{eq:setki}
\end{equation}If the state of the car fleet does not correspond to a degenerate case for which $\overline{\mc K}_{i}=\emptyset$, then every $x^{i}\in\overline{\mc K}_{i}$ is feasible and $\overline{\mc K}_{i}$ is compact and convex. 
\end{proposition}

%\medskip
\iftoggle{full_version}{
\medskip
\noindent The proof of Proposition~\ref{prop:Kbar} is given in Appendix~\ref{app:proofKbar}.
}{
\medskip
\noindent The proof of Proposition~\ref{prop:Kbar} is given in the extended version of our paper.}
\medskip
\begin{remark}
For every subset $S$ of $\mc M$, a constraint on the discrete allocation vector $n^{i}$ given by 
\begin{equation}
\sum_{j\in S}n_{j}^{i}\leq \left|\bigcup_{j \in S} \mc F_{j}^{i}\right|\,,
\label{eq:disccond}
\end{equation}
must be fulfilled so that every vehicle is matched with exactly one charging station. Intuitively, inequality \eqref{eq:disccond} states that for any subset of the charging stations, the operator of the company must not allocate more vehicles than what is feasible. In fact, the constraint on the continuous allocation vector $x^{i}$ given by \eqref{eq:setki} is a tightened version of the constraint \eqref{eq:disccond} that guarantees the condition \eqref{eq:disccond} will be fulfilled regardless of how the operator chooses $n^{i}$ based on $x^{i}$. Degenerate states of the car fleet that result in $\overline{\mc K}_{i}=\emptyset$ correspond to cases where most of the vehicles have very limited options when choosing the station to charge and as such are not the subject of our interest.    
\end{remark}
\medskip
\noindent Let $\overline{\mc K}\defineas \prod_{i\in\mc C} \overline{\mc K}_{i}$ and $\overline{\mc K}_{-i}\defineas \prod_{j\in\mc C\setminus i} \overline{\mc K}_{j}$. We will now introduce our pricing mechanism.

\medskip

\begin{definition}[System Optimal Pricing Policies]\label{def:pricingpolicies}
For each company $i \in \mc C$, let  
\begin{equation}
    p_{i}\left(x^{i},x^{-i}\right)=D_{i}^{*}\left[\frac{1}{2}\overline{\text{A}}_{i} x^{i}+ \overline{\text{B}}_{i}\sigma\left(x^{-i}\right)+\Delta_{i}\right] \,.
    \label{eq:optimalpolicy}
\end{equation}
where $\overline{\text{A}}_{i}=N_{i}^{2} A_{G}-A_{i}$ , $\overline{\text{B}}_{i}=N_{i} A_{G}-B_{i}$ and $\Delta_{i}= N_{i}b_{G}-c_{i}-f_{i}$.
\end{definition}
\medskip
\begin{remark}
For a company $i \in \mc C$, unreachable stations will correspond to zero diagonal entries in the matrix $D_i$, which makes the matrix not invertible. However, since company $i$ will not use those charging stations, letting the prices for those stations be zero through the pseudo-inverse will not affect the solution of the problem.
\end{remark}
\medskip

We will later in this section show that these pricing polices minimize the government's objective, which explains why we refer to the pricing policies as system optimal.

Next, we will show that the proposed pricing policies will give raise to a unique Nash equilibrium in the game between the companies. 

\begin{theorem}
For all companies $i \in \mc C$, let the sets $\overline{\mc K}_i$ be designed as in Proposition~\ref{prop:Kbar}. Then, with the system optimal pricing policies in Definition~\ref{def:pricingpolicies}, the game $G$ in~\eqref{eq:gameform} has a unique Nash equilibrium. 
%Let $\mc K_{i}\subseteq \mc K_{i}, \forall i \in \mc C$, be compact and convex sets that satisfy Slater's constraint qualification \cite{ConvexOptimization}. Under the pricing policies defined as 
%\begin{equation}
%    p_{i}\left(x^{i},x^{-i}\right)=C_{i}^{*}\left[\frac{1}{2}\overline{\text{A}}_{i} x^{i}+ \overline{\text{B}}_{i}\sigma\left(x^{-i}\right)+\Delta_{i}\right]
 %   \label{eq:optimalpolicy}
%\end{equation}
%where $\overline{\text{A}}_{i}=N_{i}^{2} A_{G}-A_{i}$ , $\overline{\text{B}}_{i}=N_{i} A_{G}-B_{i}$ and $\Delta_{i}= N_{i}b_{G}-d_{i}$, game $G$ has a unique Nash equilibrium.
\label{th:t2}
\end{theorem}

\begin{proof}
To prove existence and uniqueness of the Nash equilibrium, we rely on techniques from~\cite{Rosen}. Inserting policy \eqref{eq:optimalpolicy} into \eqref{eq:companyloss}, and utilizing that for $x^{i}\in\overline{\mc K}_{i}$ it holds that $D_{i}D_{i}^{*}x^{i}=x^{i}$, transforms the cost of each company $i\in\mc C$ into $J^{i} (x^i, x^{-i}) =\frac{1}{2}\left(x^{i}\right)^{T}N_{i}^{2}A_{G}x^{i}+\left(x^{i}\right)^{T}N_{i}\left(A_{G}\sigma\left(x^{-i}\right)+b_{G}\right)$.
%\begin{equation}
%    J^{i} (x^i, x^{-i}) =\frac{1}{2}\left(x^{i}\right)^{T}N_{i}^{2}A_{G}x^{i}+\left(x^{i}\right)^{T}N_{i}\left(A_{G}\sigma\left(x^{-i}\right)+b_{G}\right) \,.
%\end{equation}
Since the action spaces $\overline{\mc K}_{i}$ are compact (as a subset of the probability space over $\mc M$), convex and satisfy Slater's constraint qualification by construction, $J^{i}(x)$ are continuous in $x\in \overline{\mc K}$, $J^{i}(x^{i},x^{-i})$ are convex in $x^{i}\in\overline{\mc K}_{i}$ for a fixed $x^{-i}\in\overline{\mc K}_{-i}$ and players perform minimization of the objective, \cite[T.1]{Rosen} guarantees existence of a Nash equilibrium. According to \cite[T.2]{Rosen}, a sufficient condition for the Nash equilibrium to be unique is that the symmetric matrix $\Gamma\defineas G(x,r)+G^{T}(x,r)$ be negative definite for $x\in\overline{\mc K}$ and some $r=\left[r_{i}\right]_{i\in\mc C}\in \R^{\mc C}_{>0}$ with $G(x,r)$ being the Jacobian with respect to $x$ of function $g(x,r)$ defined as $g(x,r)\defineas \left[-r_{i}\nabla_{x_{i}}J^{i}(x)\right]_{i\in \mc C} $. For $r=\textbf{1}_{|\mc C|}$ and any $x\in \overline{\mc K}$ we have
$
    x^{T}\Gamma x=-2\left(\sum_{i\in \mc C}N_{i}x^{i}\right)^{T}A_{G}\left(\sum_{i\in \mc C}N_{i}x^{i}\right)$.

Since $\overline{\mc K}_{i}\subseteq\mc P_{\mc M}$ for all $i \in \mc C$, we have $\sum_{i\in \mc C}N_{i}x^{i}\neq \textbf{0}_{|\mc M|}$. Since $A_{G}\succ 0$, we have  $x^{T}\Gamma x<0$ for all $ x\in\overline{\mc K}$ which proves that $\Gamma$ is negative definite on $\overline{\mc K}$ and that the Nash equilibrium is unique.
\end{proof}

\medskip

Now that we know that under the pricing policies given by \eqref{eq:optimalpolicy} the Nash equilibrium is unique, we proceed to show that it also minimizes the government objective \eqref{eq:JG}.
\begin{theorem}
For all companies $i \in \mc C$, let the sets $\overline{\mc K}_i$ be designed as in Proposition~\ref{prop:Kbar}. Then, with the system optimal pricing policies in Definition~\ref{def:pricingpolicies}, the Nash equilibrium $x^{*}$ of game $G$ satisfies
\begin{equation}
    x^{*} \in \argmin_{x \in \overline{\mc K}} J_{G}\left(\sigma(x)\right) \,. 
\end{equation}
\end{theorem}

\begin{proof}
Let $A^{T}\defineas\left[N_{i}\mathbb{I}_{|\mc M|}\right]_{i\in \mc C}\in\R^{|\mc M||\mc C|\times |\mc M|}$, then the government optimization problem is equivalent to 
\begin{equation}
    \min_{x\in \overline{\mc K}}\:J_{G}(x)\defineas\frac{1}{2}x^{T}A^{T}A_{G}Ax+b_{G}^{T}Ax \,.
    \label{eq:JGx}
\end{equation}
The function $J_{G}(x)$ is convex since $\nabla^{2}_{x}J_{G}(x)=A^{T}A_{G}A$ and for all $x^{i} \in \overline{\mc K}_{i}\subseteq\mc P_{\mc M}$ it holds that $\sum_{i\in \mc C}N_{i}x^{i}\neq \textbf{0}_{|\mc M|}$ so $x^{T}A^{T}A_{G}Ax=\left(\sum_{i\in\mc C}N_{i}x^{i}\right)^{T}A_{G}\left(\sum_{i\in\mc C}N_{i}x^{i}\right)> 0$ which guarantees that $\nabla^{2}_{x}J_{G}(x)\succeq 0$. According to \cite[4.21]{ConvexOptimization}, $x^{*}$ is the minimizer of \eqref{eq:JGx} on $\overline{\mc K}$ if and only if $\left\langle\nabla_{x} J_{G}(x)\mid_{x=x^{*}}, y-x^{*}\right\rangle \geq 0$, $\forall y \in \overline{\mc K}$.
Under the pricing policies defined in \eqref{eq:optimalpolicy}, $J_{G}(x)$ is the exact potential \cite{PotentialGames} for game $G$ satisfying for all $ i\in\mc C$ and any fixed $ x^{-i}\in\overline{\mc K}_{-i}$
\begin{equation}
    \nabla_{x^{i}}J^{i}\left(x^{i},x^{-i}\right)=\nabla_{x^{i}}J_{G}\left(x^{i},x^{-i}\right) \,, \quad \forall x^{i}\in \overline{\mc K}_{i} \,.
    \label{eq:potential}
\end{equation}
If $\widehat{x}$ is the Nash equilibrium of $G$, then for all $ i\in\mc C$ we have $\widehat{x}^{i}\in\argmin_{x^{i}\in\overline{\mc K}_{i}}J^{i}\left(x^{i},\widehat{x}^{-i}\right)$. According to \cite{ConvexOptimization}, we can now write $\left\langle\nabla_{x^{i}} J^{i}(x)\mid_{x=\widehat{x}}, y^{i}-\widehat{x}^{i}\right\rangle \geq 0$, for all $ y^{i}\in\overline{\mc K}_{i}$ for all $i\in\mc C$. Because of \eqref{eq:potential}, for all $i\in\mc C$ and for all $ y^{i}\in\overline{\mc K}_{i}$ it holds that 
$
    \sum_{i\in\mc C}\left\langle\nabla_{x^{i}} J_{G}(x)\mid_{x=\widehat{x}}, y^{i}-\widehat{x}^{i}\right\rangle \geq 0 \,.
$
Finally, we have that $\widehat{x}$ indeed is the minimizer of \eqref{eq:JGx} since $\sum_{i\in\mc C}\left\langle\nabla_{x^{i}} J_{G}(x)\mid_{x=\widehat{x}}, y^{i}-\widehat{x}^{i}\right\rangle=\left\langle\nabla_{x} J_{G}(x)\mid_{x=\widehat{x}}, y-\widehat{x}\right\rangle$ is true for any $y\in\overline{\mc K}$.
\end{proof}

\medskip

Sets $\overline{\mc K}_{i}$ defined as \eqref{eq:setki} reflect the current state of the car fleets. In realistic scenarios, these sets are private, i.e., not known to the government, as they encompass information about the current true location of the vehicles and their current and desired battery status, preventing centralized computation of the Nash equilibrium. Hence, a decentralized algorithm with minimal exchange of information between the agents is required. Such algorithms based on theory of aggregative games were proposed in \cite{Paccagnan2016a}, \cite{Paccagnan2016b} and \cite{Decentralized}. Based on \cite{Decentralized}, since our game-map, defined as $F(x)=\left[\nabla_{x^{i}}J^{i}\left(x^{i},x^{-i}\right)\right]_{i\in\mc C}$, is equal to 
\begin{equation}
    F(x)=\underbrace{A^{T}A_{G}A}_{F_{1}} x+\underbrace{A^{T}b_{G}}_{F_{2}}
    \label{eq:gamemap}
\end{equation}
and is a non-strictly monotonic ($F_{1}\succeq 0$) linear operator, we utilize a distributed iterative scheme based on the Krasnoselskij iteration \cite{ApproxFixPoint} to find the Nash equilibrium of $G$.
\begin{proposition}
Under the system optimal pricing policies and for sets $\overline{\mc K}_i$ as in Proposition~\ref{prop:Kbar}, for every $\gamma$ such that
\begin{equation}
    0<\gamma<\frac{2}{\lambda_{\text{max}}\left(F_1\right)} \,,
    \label{eq:optgamma}
\end{equation}
a distributed iterative scheme given by
\begin{equation}
    x^{i}(k+1)=\frac{1}{2}\left(x^{i}(k)+\Pi_{\overline{\mc K}_{i}}\left[x^{i}(k)-\gamma \nabla_{x^{i}} J^{i}\left(x^{i}, x^{-i}\right)\right]\right) \,, 
\label{eq:iterativeproc}
\end{equation}
where $\Pi_{\overline{\mc K}_{i}}$ denotes the projection operator onto $\overline{\mc K}_{i}$, converges to the Nash equilibrium of the game $G$. 
\end{proposition}

\begin{proof}
A point $\widehat{x}\in\overline{\mc K}$ is a Nash equilibrium of game $G$ with game map $F(x)$ defined by \eqref{eq:gamemap} if and only if $F\left(\widehat{x}\right)^{T}\left(y-\widehat{x}\right)\geq 0$ for all $ y \in \overline{\mc K}$ \cite{VISurvey}. One can prove that $F\left(\widehat{x}\right)^{T}\left(y-\widehat{x}\right)\geq 0$ holds for all $y\in\overline{\mc K}$ if and only if $\widehat{x}=\Pi_{\overline{\mc K}}\left[\widehat{x}-\gamma F\left(\widehat{x}\right)\right]$. Indeed, based on~\cite{ConvexOptimization} and the fact that
$z=\Pi_{\overline{\mc K}}\left[\widehat{x}-\gamma F\left(\widehat{x}\right)\right]$ is equivalent to $z$ being the minimizer of $\norm{z-\left(\widehat{x}-\gamma F(\widehat{x})\right)}^{2}_{2}$ over $\overline{\mc K}$, we have that $z=\Pi_{\overline{\mc K}}\left[\widehat{x}-\gamma F\left(\widehat{x}\right)\right]$ is equivavlent to $ 2\left(z-\left(\widehat{x}-\gamma F(\widehat{x})\right)\right)^{T}(y-z)\geq 0$, for all $y\in\overline{\mc K}$, which by setting $z=\widehat{x}$ completes the proof of equivalence. Now we have that $\widehat{x}$ is a Nash equilibrium if and only if it is a fixed point of $H(x)=\Pi_{\overline{\mc K}}\left[x-\gamma F\left(x\right)\right]$. Because $\overline{\mc K}$ is compact and convex, for any $\gamma$ such that $H(x)$ is non-expansive  and $x(0)\in\overline{\mc K}$, the iterative procedure $x(k+1)=0.5\left(x(k)+H\left(x(k)\right)\right)$ converges to the fixed point of $H(x)$ (the unique Nash equilibrium of $G$) according to \cite{ApproxFixPoint}. Since the projection operator is non-expansive, for $\gamma$ such that $\bar{H}(x)=\mathbb{I}x-\gamma F(x)=\left(\mathbb{I}-\gamma F_{1}\right)x-\gamma F_{2}$ is non-expansive, the map $H(x)$ will be non-expansive  too. $\bar{H}(x)$ is an affine map so it is non-expansive if $\norm{\mathbb{I}-\gamma F_{1}}_{2}\leq1$. Since $F_{1}$ is symmetric this is equivalent to $\max_{i}\left|\lambda_{i}\left(\mathbb{I}-\gamma F_{1}\right)\right|\leq1$. For $\gamma$ given in \eqref{eq:optgamma} and because $F_{1}\succeq 0$ this is guaranteed since 
$-1\leq 1-\gamma\lambda_{i}\left( F_{1}\right)\leq 1$ is true for all $i$.
\end{proof}

\section{Numerical Example}\label{sec:example}

We illustrate in this section how the proposed method can be utilized to balance the EVs so that the number of them charging at different stations is as close as possible to vector $\hat{N}$. We consider a scenario where 3 ride-hailing companies $\mc C=\{\mc C_1,\mc C_2,\mc C_3\}$, whose fleet sizes are given by $N=\left[60,35,45\right]^{T}$, operate in a square region with 4 charging stations $\mc M=\{\mc M_1,\mc M_2,\mc M_3,\mc M_4\}$. The stations are described by the vector of their capacities $M=\left[20, 10, 15, 10\right]^{T}$ and we set desired vehicle numbers around them to be $\hat{N}=\left[35,15,50,40\right]^{T}$.

Each vehicle $v_{j} \in V_{i}$ is described by a tuple $\left(x_{j}, y_{j}, s_{j}^{\text{start}}, s_{j}^{\text{des}}, d_{j}^{\text{max}}\right)$ where $(x_{j}, y_{j})\in \mathbb{R}^{2}$ describes the position of the vehicle, $d_{j}^{\text{max}}$ is the max range of the vehicle and $s_{j}^{\text{start}}, s_{j}^{\text{des}}$ represent the current and desired battery levels. The vehicles and charging stations are placed randomly with: $s^{\text{start}}\sim \mathcal{U}[20, 40]$, $s^{\text{des}}\sim \mathcal{U}[80, 100]$ and $d^{\text{max}}\sim \mathcal{U}[150, 200]$. The scenario is depicted in Figure \ref{fig:Map}.
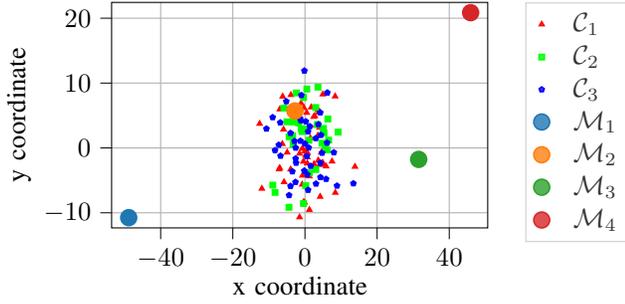
\begin{figure}
    \centering
    \input{figures/Map.tikz}
    \caption{Shows the location of charging station 1 $\left(\mc M_{1}\right)$, station 2 $\left(\mc M_{2}\right)$, station 3 $\left(\mc M_{3}\right)$ and station 4 $\left(\mc M_{4}\right)$. Moreover, black triangles, squares and pentagons show the locations of vehicles that belong to company 1 $\left(\mc C_{1}\right)$, company 2 $\left(\mc C_{2}\right)$ and company 3 $\left(\mc C_{3}\right)$ respectively.}
    \label{fig:Map}
\end{figure} A station is considered to be feasible to a vehicle if the vehicle can reach it with the current battery status. For simplicity, if we assume a linear battery discharge model, a charging station $k$ is feasible for vehicle $j$ if $s_{j}^{\text{start}}-\frac{100}{d_{j}^{\text{max}}}d_{j,k}>0$ where $d_{j,k}$ denotes the distance between the vehicle $j$ and the charging station $k$ and $s_{j}^{\text{start}}$ is expressed in percentage. The average charging cost is modelled as $J_{2}^{i}\left(x^{i}, p_{i}\left(x^{i},x^{-i}\right)\right)=N_{i}\left(x^{i}\right)^{T}R_{i} p_{i}\left(x^{i},x^{-i}\right).$ Diagonal matrix $R_{i}\in \mathbb{R}^{4\times 4}$ captures the average charging demand per vehicle when choosing each of the charging stations. For infeasible charging stations the average demand is set to 0. Pricing policy $p_{i}$ denotes the price of one unit of charge at each station. If the charging station $k$ is feasible to vehicle $v_{l}\in V_{i}$, vehicle's charging demand if $k$ is chosen for charging is defined as $\delta_{l,k}=\beta_{l}\left(s_{l}^{\text{des}}-\left(s_{l}^{\text{start}}-\frac{100}{d_{l}^{\text{max}}}d_{lk}\right)\right)$.
Here $\beta_{l}\in \mathbb{R}$ is a scaling coefficient that says how many units of charge corresponds to $1\%$ of the vehicle's battery. The diagonal element of $R_{i}$ that corresponds to station $k$ is then given by $(R_{i})_{kk}=\left[\frac{1}{\left|\mc F_{k}^{i}\right|}\sum_{l:v_{l}\in \mc F_{k}^{i}}\delta_{l,k}\right]$. We model the negative expected revenue as $J_{3}^{i}\left(x^{i}\right)=\left(e_{i}^{\text{arr}}\right)^{T}N_{i}x^{i}-\left(e_{i}^{\text{pro}}\right)^{T}N_{i}x^{i}.$ Here, $e_{i}^{\text{arr}}\in \mathbb{R}^{4}$ is the average cost of a vehicle being unoccupied while traveling to a charging station. If station $k$ is infeasible, then we set $\left(e_{i}^{\text{arr}}\right)_{k}=0$, otherwise it is equal to $\left(e_{i}^{\text{arr}}\right)_{k}=u_{i}\cdot P_{k}\cdot\left[\frac{1}{\left|\mc F_{k}^{i}\right|}\sum_{l:v_{l}\in \mc F_{k}^{i}}d_{l,k}\right]$, where $u_{i}\in \mathbb{R}$ is the monetary value of a vehicle being occupied while driving for $1 \text{km}$, given in $\left[\$/\text{km}\right]$ and $P_{k}$ is the probability of a vehicle being occupied in the region around charging station $k$. The vector $e_{i}^{\text{pro}}\in \mathbb{R}^{4}$ denotes expected profit in regions around different charging stations. In general, this vector is obtained from historical data and here we choose it randomly such that each element of $e^{\text{pro}}$ satisfies $e^{\text{pro}}_{j}\sim \mathcal{U}[100,350]$. The sample drawn in this simulation is $e^{\text{pro}}=[202.51, 301.02, 252.34, 195.61]^{T}$. We fix other parameters to $\beta_{l}=1.0$, $Q=\text{diag}(1, 5, 3, 2)$ and $A_{G}=2Q$, vector of probabilities of being occupied $P=[0.15, 0.4, 0.2, 0.1]$ for all $k\in\mc M$, $u_{i}=1.0$ for all $i\in\mc C$, and set the number of iterations for the algorithm to $k=3000$. For this case study, the optimal pricing policy in accordance with \eqref{eq:optimalpolicy} is obtained by setting $D_{i}\defineas N_{i}R_{i}$ and $f_{i}\defineas N_{i}\left(e^{\text{arr}}-e^{\text{pro}}\right)$.

In the Nash equilibrium, car fleet portions to be directed to each of the charging stations and the resulting charging prices are presented in Table \ref{tab:res} whereas the evolution of the government loss $J_{G}$ and the total number of vehicles over the iterations is presented in Figure \ref{fig:results}.
\begin{table}
\begin{center}
 \renewcommand{\arraystretch}{1.2}
 \caption{Company decisions and charging prices}\vspace{1ex}
 \label{tab:res}
 \begin{tabular}{c|cc|cc|cc|cc}
 %\hline
 \multirow{2}{*}{$\mc C$} & \multicolumn{2}{c|}{Station 1} & \multicolumn{2}{c|}{Station 2} & \multicolumn{2}{c|}{Station 3} & \multicolumn{2}{c}{Station 4} \\ %\cline{2-9}
 & $x_{1}^{i}$ & $p_{i}$ &$x_{2}^{i}$ & $p_{i}$ & $x_{3}^{i}$ & $p_{i}$ & $x_{4}^{i}$ & $p_{i}$ \\
 \hline
 $\mc C_{1} \rule{0pt}{2.6ex}$ &  0.20 & 1.65  & 0.15 & 3.78  & 0.38 & 1.16  & 0.27 & 0.98  \\
 $\mc C_{2}$ &  0.19 & 1.75  & 0.16 & 4.12  & 0.41 & 1.48  & 0.24 & 1.17 \\
 $\mc C_{3}$ &  0.21 & 1.77  & 0.10 & 4.16  & 0.43 & 1.29  & 0.26 & 1.04 \\
\hline 
\end{tabular}
\end{center}
\end{table}
\begin{figure}
    \centering
    \input{figures/JgSigma.tikz}
    \caption{The evolution of the government loss function and the total number of vehicles to be charged at each station during the iterative procedure for finding the Nash equilibrium. Dotted lines represent the target values. The execution time of all simulations on an average PC was smaller than 3 sec.}
    \label{fig:results}
\end{figure}
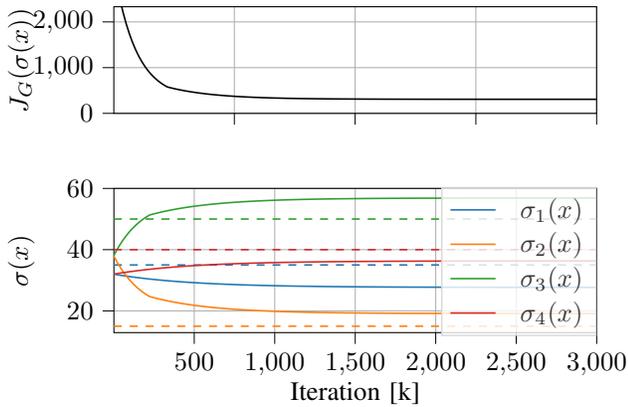
From the plot it is clear that the iterative procedure converged to a Nash equilibrium that is the government optimum but does not perfectly match the predefined vehicle accumulation vector $\hat{N}$ due to vehicle arrangement and their battery status. As expected, the prices of charging at station 2 are significantly higher than for any other charging station for all the companies as it has the smallest desired vehicle accumulation and is the most desirable in terms of expected profit and the distance to be travelled to reach it. Station 4 is the least attractive hence, it has the smallest charging prices in the Nash equilibrium. 

Apart from $R_{i}$ and $e^{\text{arr}}$, all other parameters are inherently known to the government as they characterize the region in which the companies operate. Hence, the government optimum is attainable if the companies are willing to share $R_{i}$ and $e^{\text{arr}}$ that encompass the information about the average state of the company's fleet. We test robustness of the proposed pricing policies and show how the system behaves in the same scenario when the government has only an estimate $\overline{R}_{i}$ of the average charging demand $R_{i}$. For a feasible station $k$, we let $\left(\overline{R}_{i}\right)_{kk}=\left(R_{i}\right)_{kk}+w_{k}$ where $w_{k}$ is a noise sample such that $w_{k}\sim \mathcal{N}\left(0,\left(\alpha R_{\text{min}}/5\right)^2\right)$ with $R_{\text{min}}$ being the minimal, non-zero, diagonal element of $R_{i}$. For every $\alpha$ we sample $w_{k}$ one hundred times and report the mean value of the government's loss in the Nash equilibrium. Figure \ref{fig:robustness} shows that for moderate discrepancies $(\alpha<0.6)$ between the true and the estimated value of $R_{i}$, the attained Nash equilibrium is close to the government's optimum. It also confirms that the worse the approximation is, the higher the deviation of the Nash equilibrium from $\hat{N}$ will be. 
\begin{figure}
    \centering
    \input{figures/robustness.tikz}

    \caption{When the government does not have correct information about the vehicles' locations and what their charging demands are, the pricing policies divert from the optimal ones. For each $\alpha$, the mean value of $J_G$ is plotted over 100 simulations, together with its maximum and minimum value. }
    \label{fig:robustness}
\end{figure}
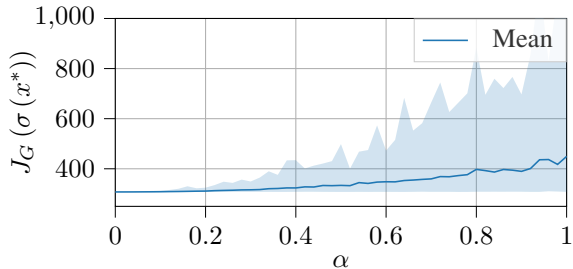

\section{Conclusions}\label{sec:conclusion}

In this paper we have developed a model for charge pricing of fleets of electric ride-hailing vehicles, where a central authority wants to control the demand on the charging stations through pricing. We constructed a set of pricing policies, and showed that those policies both give rise to a unique Nash equilibrium when each fleet operator wants to minimize its own operational cost and that this Nash equilibrium also minimizes the deviation from the central authority's desire.

In the future, we plan to deeper address the robustness of the proposed solution, something that is needed when the government does not have full knowledge of vehicles' position and charging demands.

\bibliographystyle{IEEEtran}
\bibliography{references.bib}

\iftoggle{full_version}{
\appendix
\subsection{Proof of Proposition~\ref{prop:Kbar}}\label{app:proofKbar}
\begin{proof}
We start by showing that it is possible to match each vehicle $v\in \mc V_{i}$ with exactly one charging station if and only if for all $S \subseteq \mc M$ the equation \eqref{eq:disccond} holds. To prove this, we look at a bipartite graph $G_{i}=\left(\mc V_{i}\cup \mc S_{i} , E_{i}\right)$ where $\mc S_{i}$ is defined as $\mc S_{i}=\bigcup_{j\in \mc M}\mc S_{i}^{j}$ such that for all $j_{1},j_{2}\in \mc M$, $j_{1}\neq j_{2}$ it holds that $\mc S_{i}^{j_{1}}\cap \mc S_{i}^{j_{2}}=\emptyset$. Each $\mc S_{i}^{j}$ is comprised of $n_{j}^{i}$ copies of the vertex that corresponds to the charging station $j$. The set of edges $E_{i}$ is formed such that $v\in \mc V_{i}$ is connected to  $s\in \mc S_{i}^{j}$ if $v\in \mc F_{j}^{i}$. The two sets have equal number of vertices $\left|\mc V_{i}\right|=N_{i}=\sum_{j\in\mc M}n^{i}_{j}=\left|\mc S_{i}\right|$ which means that desired matching is possible if and only if there exists an $\mc S_{i}-$perfect matching on graph $G_{i}$. Since condition \eqref{eq:disccond} corresponds exactly to the condition of the Hall's marriage theorem, the equivalence is proved. We now show that if $x^{i}\in\overline{\mc K}_{i}$ defined in Proposition~\ref{prop:Kbar} then $x^{i}$ is feasible. This means that $n_{j}^{i}$ defined according to Section \ref{sec:model}, satisfies the assumption given by \eqref{eq:disccond} for    $x^{i}\in\overline{\mc K}_{i}$ defined in Proposition \ref{prop:Kbar}. We distinguish 2 cases: $S\subset \mc M$ and $S=\mc M$. For $S\subset \mc M$ we can write $$\sum_{j\in S}n_{j}^{i}=\sum_{j\in P_{1}}\left\lfloor{N_{i}x^{i}_{j}}\right\rfloor + \sum_{j\in P_{2}}\left\lceil{N_{i}x^{i}_{j}}\right\rceil$$ 
where $P_{1}\cup P_{2}=S\wedge P_{1}\cap P_{2}=\emptyset$. We have
$\sum_{j\in P_{1}}\left\lfloor{N_{i}x^{i}_{j}}\right\rfloor\leq \sum_{j\in P_{1}}N_{i}x^{i}_{j}$ and
$\sum_{j\in P_{2}}\left\lceil{N_{i}x^{i}_{j}}\right\rceil= \sum_{j\in P_{2}}N_{i}x^{i}_{j}+\left\{N_{i}x_{j}^{i}\right\}$
where $\forall j \in P_{2}$ it holds that $\left\{N_{i}x_{j}^{i}\right\}\leq 1$. We have $$\sum_{j\in S}n_{j}^{i}\leq \sum_{j\in P_{1}\cup P_{2}}N_{i}x^{i}_{j}+\sum_{j\in P_{2}}\left\{N_{i}x_{j}^{i}\right\}\leq \sum_{j\in S}N_{i}x^{i}_{j}+\left|P_{2}\right|$$ which combined with \eqref{eq:setki} finally gives
$$\sum_{j\in S}n_{j}^{i}\leq \left|\bigcup_{j \in S} \mc F_{j}^{i}\right| - \left|S\right|+\left|P_{2}\right|\leq \left|\bigcup_{j \in S}\mc F_{j}^{i}\right|$$ because $\left|P_{2}\right|\leq \left|S\right|$. For $S=\mc M$ we have that the condition given by \eqref{eq:disccond} is fulfilled with the equality since $\sum_{j\in S}n_{j}^{i} = N_{i}=\left|\bigcup_{j \in S}\mc F_{j}^{i}\right|$. The case when for some $S$ it holds that $\left|\bigcup_{j \in S} \mc F_{j}^{i}\right| - \left|S\right|\leq 0$ leads to $x_{j}^{i}=0$ for all $j\in S$, which in return leads to $n_{j}^{i}=0$, so no matching is required in that case. By construction, 
sets $\mc K_{i}$ are defined as the intersection of a probability space and $2^{m}-2$ linear inequalities given by \eqref{eq:setki}, making them compact and convex. 
\end{proof}
}

\end{document}

%% file: figures/Map.tikz
% This file was created with tikzplotlib v0.9.15.
\begin{tikzpicture}

\definecolor{color0}{rgb}{0.12156862745098,0.466666666666667,0.705882352941177}
\definecolor{color1}{rgb}{1,0.498039215686275,0.0549019607843137}
\definecolor{color2}{rgb}{0.172549019607843,0.627450980392157,0.172549019607843}
\definecolor{color3}{rgb}{0.83921568627451,0.152941176470588,0.156862745098039}

\definecolor{markcolor1}{rgb}{1,0,0}
\definecolor{markcolor2}{rgb}{0,1,0}
\definecolor{markcolor3}{rgb}{0,0,1}

\begin{axis}[
legend cell align={left},
legend style={
  fill opacity=0.8,
  draw opacity=1,
  text opacity=1,
  at={(1.1,1.0)},
  anchor=north west,
  draw=white!80!black
},
tick align=outside,
tick pos=left,
x grid style={white!69.0196078431373!black},
xlabel={x coordinate},
xmajorgrids,
xmin=-53.5919005658097, xmax=50.6626807644641,
xtick style={color=black},
y grid style={white!69.0196078431373!black},
ylabel={y coordinate},
ymajorgrids,
ymin=-12.3286965650359, ymax=22.4359791136868,
ytick style={color=black},
width=5cm,
height=3cm,
scale only axis
]
\addplot [semithick, mark=triangle*, markcolor1, mark size=1, mark options={solid}, only marks]
table {%
-6.27653362858984 0.624103784974513
};
\addlegendentry{$\mc C_{1}$}
\addplot [semithick, mark=triangle*, markcolor1, mark size=1, mark options={solid}, forget plot]
table {%
-0.740340656456942 -0.100257223757533
};
\addplot [semithick, mark=triangle*, markcolor1, mark size=1, mark options={solid}, forget plot]
table {%
4.94026672205929 8.27042044945629
};
\addplot [semithick, mark=triangle*, markcolor1, mark size=1, mark options={solid}, forget plot]
table {%
1.22916227058458 -9.53749238036561
};
\addplot [semithick, mark=triangle*, markcolor1, mark size=1, mark options={solid}, forget plot]
table {%
8.36492244560831 7.96000394312338
};
\addplot [semithick, mark=triangle*, markcolor1, mark size=1, mark options={solid}, forget plot]
table {%
13.865800766076 -2.88561644555743
};
\addplot [semithick, mark=triangle*, markcolor1, mark size=1, mark options={solid}, forget plot]
table {%
-0.129482068359837 -8.38369899092729
};
\addplot [semithick, mark=triangle*, markcolor1, mark size=1, mark options={solid}, forget plot]
table {%
-7.17112366730246 -3.28285662684277
};
\addplot [semithick, mark=triangle*, markcolor1, mark size=1, mark options={solid}, forget plot]
table {%
2.71881843146671 4.92696074114078
};
\addplot [semithick, mark=triangle*, markcolor1, mark size=1, mark options={solid}, forget plot]
table {%
0.228606766245491 -1.13806883032336
};
\addplot [semithick, mark=triangle*, markcolor1, mark size=1, mark options={solid}, forget plot]
table {%
-1.49895454280914 -10.6761830128812
};
\addplot [semithick, mark=triangle*, markcolor1, mark size=1, mark options={solid}, forget plot]
table {%
0.649985063682959 -1.43621782562308
};
\addplot [semithick, mark=triangle*, markcolor1, mark size=1, mark options={solid}, forget plot]
table {%
-0.0665540739414045 -0.507353541443753
};
\addplot [semithick, mark=triangle*, markcolor1, mark size=1, mark options={solid}, forget plot]
table {%
-0.76504514038511 -5.633390914266
};
\addplot [semithick, mark=triangle*, markcolor1, mark size=1, mark options={solid}, forget plot]
table {%
5.85796096087268 -2.83094149600252
};
\addplot [semithick, mark=triangle*, markcolor1, mark size=1, mark options={solid}, forget plot]
table {%
7.50337982858862 -2.07200983312489
};
\addplot [semithick, mark=triangle*, markcolor1, mark size=1, mark options={solid}, forget plot]
table {%
0.488645763112178 5.42757530350688
};
\addplot [semithick, mark=triangle*, markcolor1, mark size=1, mark options={solid}, forget plot]
table {%
-0.271813887538355 6.1752601385833
};
\addplot [semithick, mark=triangle*, markcolor1, mark size=1, mark options={solid}, forget plot]
table {%
-5.94510340607541 2.84702183230961
};
\addplot [semithick, mark=triangle*, markcolor1, mark size=1, mark options={solid}, forget plot]
table {%
-0.689393999593766 -4.24966868116763
};
\addplot [semithick, mark=triangle*, markcolor1, mark size=1, mark options={solid}, forget plot]
table {%
0.68879114843478 -0.740024019176562
};
\addplot [semithick, mark=triangle*, markcolor1, mark size=1, mark options={solid}, forget plot]
table {%
-11.9474524197219 -6.27035932303288
};
\addplot [semithick, mark=triangle*, markcolor1, mark size=1, mark options={solid}, forget plot]
table {%
-0.415534519735201 -8.36498268279421
};
\addplot [semithick, mark=triangle*, markcolor1, mark size=1, mark options={solid}, forget plot]
table {%
-3.97196835324742 -3.30909540094792
};
\addplot [semithick, mark=triangle*, markcolor1, mark size=1, mark options={solid}, forget plot]
table {%
3.85600726975457 5.50723822655677
};
\addplot [semithick, mark=triangle*, markcolor1, mark size=1, mark options={solid}, forget plot]
table {%
-3.97225536416136 5.53165487052545
};
\addplot [semithick, mark=triangle*, markcolor1, mark size=1, mark options={solid}, forget plot]
table {%
1.75063371497702 6.27586659236018
};
\addplot [semithick, mark=triangle*, markcolor1, mark size=1, mark options={solid}, forget plot]
table {%
0.531980168713532 -0.717784285684946
};
\addplot [semithick, mark=triangle*, markcolor1, mark size=1, mark options={solid}, forget plot]
table {%
-5.79229561375719 -5.2572846722155
};
\addplot [semithick, mark=triangle*, markcolor1, mark size=1, mark options={solid}, forget plot]
table {%
6.20692290971525 1.0246789364199
};
\addplot [semithick, mark=triangle*, markcolor1, mark size=1, mark options={solid}, forget plot]
table {%
-0.448984145052607 -2.09448213248733
};
\addplot [semithick, mark=triangle*, markcolor1, mark size=1, mark options={solid}, forget plot]
table {%
3.60373448872801 5.86586996166557
};
\addplot [semithick, mark=triangle*, markcolor1, mark size=1, mark options={solid}, forget plot]
table {%
4.94092608530561 -2.57239426489411
};
\addplot [semithick, mark=triangle*, markcolor1, mark size=1, mark options={solid}, forget plot]
table {%
0.644275720692194 -2.35991453955683
};
\addplot [semithick, mark=triangle*, markcolor1, mark size=1, mark options={solid}, forget plot]
table {%
4.90793739440746 -2.26124685677002
};
\addplot [semithick, mark=triangle*, markcolor1, mark size=1, mark options={solid}, forget plot]
table {%
8.41824533714522 -6.89904092814916
};
\addplot [semithick, mark=triangle*, markcolor1, mark size=1, mark options={solid}, forget plot]
table {%
4.12692336908872 -2.57245525729997
};
\addplot [semithick, mark=triangle*, markcolor1, mark size=1, mark options={solid}, forget plot]
table {%
0.371413365287126 2.6883895016628
};
\addplot [semithick, mark=triangle*, markcolor1, mark size=1, mark options={solid}, forget plot]
table {%
1.49062269591157 1.52218224484258
};
\addplot [semithick, mark=triangle*, markcolor1, mark size=1, mark options={solid}, forget plot]
table {%
4.05783411732039 3.14651666067427
};
\addplot [semithick, mark=triangle*, markcolor1, mark size=1, mark options={solid}, forget plot]
table {%
3.57930571708039 1.91085626560067
};
\addplot [semithick, mark=triangle*, markcolor1, mark size=1, mark options={solid}, forget plot]
table {%
-1.49979875269698 -0.86122549590259
};
\addplot [semithick, mark=triangle*, markcolor1, mark size=1, mark options={solid}, forget plot]
table {%
3.68157573364635 -1.91295772922564
};
\addplot [semithick, mark=triangle*, markcolor1, mark size=1, mark options={solid}, forget plot]
table {%
4.21004965782807 -7.51302528244845
};
\addplot [semithick, mark=triangle*, markcolor1, mark size=1, mark options={solid}, forget plot]
table {%
-6.99948798655952 -3.07566195595901
};
\addplot [semithick, mark=triangle*, markcolor1, mark size=1, mark options={solid}, forget plot]
table {%
1.208937757926 -4.36032143895261
};
\addplot [semithick, mark=triangle*, markcolor1, mark size=1, mark options={solid}, forget plot]
table {%
-5.30014418562885 6.15813901651744
};
\addplot [semithick, mark=triangle*, markcolor1, mark size=1, mark options={solid}, forget plot]
table {%
4.53032549963734 -0.689339248523625
};
\addplot [semithick, mark=triangle*, markcolor1, mark size=1, mark options={solid}, forget plot]
table {%
5.98898469852934 -0.297165648343704
};
\addplot [semithick, mark=triangle*, markcolor1, mark size=1, mark options={solid}, forget plot]
table {%
-12.5641959045437 3.7366431261058
};
\addplot [semithick, mark=triangle*, markcolor1, mark size=1, mark options={solid}, forget plot]
table {%
-1.08066807788419 6.90507100240052
};
\addplot [semithick, mark=triangle*, markcolor1, mark size=1, mark options={solid}, forget plot]
table {%
-6.18112965330078 7.93494912887678
};
\addplot [semithick, mark=triangle*, markcolor1, mark size=1, mark options={solid}, forget plot]
table {%
3.10979029001086 3.46749267232243
};
\addplot [semithick, mark=triangle*, markcolor1, mark size=1, mark options={solid}, forget plot]
table {%
4.08085895364019 5.56125625884388
};
\addplot [semithick, mark=triangle*, markcolor1, mark size=1, mark options={solid}, forget plot]
table {%
2.52801690267689 -2.46869336408228
};
\addplot [semithick, mark=triangle*, markcolor1, mark size=1, mark options={solid}, forget plot]
table {%
-4.93062996500568 -0.695110607222465
};
\addplot [semithick, mark=triangle*, markcolor1, mark size=1, mark options={solid}, forget plot]
table {%
-3.20342536596235 0.889781767725908
};
\addplot [semithick, mark=triangle*, markcolor1, mark size=1, mark options={solid}, forget plot]
table {%
-3.89694485028449 8.17715368072877
};
\addplot [semithick, mark=triangle*, markcolor1, mark size=1, mark options={solid}, forget plot]
table {%
2.42263367105759 4.80561989129838
};
\addplot [semithick, mark=triangle*, markcolor1, mark size=1, mark options={solid}, forget plot]
table {%
-6.97342347369411 5.96322188653511
};
\addplot [semithick, mark=square*, markcolor2, mark size=1, mark options={solid}, only marks]
table {%
0.220331319522192 -5.82456858747733
};
\addlegendentry{$\mc C_{2}$}
\addplot [semithick, mark=square*, markcolor2, mark size=1, mark options={solid}, forget plot]
table {%
6.559530443756 1.20415586790711
};
\addplot [semithick, mark=square*, markcolor2, mark size=1, mark options={solid}, forget plot]
table {%
-4.40219498666908 1.71536731298668
};
\addplot [semithick, mark=square*, markcolor2, mark size=1, mark options={solid}, forget plot]
table {%
9.19174162929033 2.43461840022162
};
\addplot [semithick, mark=square*, markcolor2, mark size=1, mark options={solid}, forget plot]
table {%
6.15859949845285 -0.309772233148254
};
\addplot [semithick, mark=square*, markcolor2, mark size=1, mark options={solid}, forget plot]
table {%
0.348518376041156 2.88319466299645
};
\addplot [semithick, mark=square*, markcolor2, mark size=1, mark options={solid}, forget plot]
table {%
-8.28993902057161 -6.85680965856805
};
\addplot [semithick, mark=square*, markcolor2, mark size=1, mark options={solid}, forget plot]
table {%
-2.3588399149099 8.47966067207886
};
\addplot [semithick, mark=square*, markcolor2, mark size=1, mark options={solid}, forget plot]
table {%
-0.383936364984927 7.76843131832177
};
\addplot [semithick, mark=square*, markcolor2, mark size=1, mark options={solid}, forget plot]
table {%
-4.42733312187486 -9.1543148823871
};
\addplot [semithick, mark=square*, markcolor2, mark size=1, mark options={solid}, forget plot]
table {%
-3.42196162786722 1.8771570078365
};
\addplot [semithick, mark=square*, markcolor2, mark size=1, mark options={solid}, forget plot]
table {%
0.209958842866694 2.6973963285977
};
\addplot [semithick, mark=square*, markcolor2, mark size=1, mark options={solid}, forget plot]
table {%
5.5691433124004 0.258725141802981
};
\addplot [semithick, mark=square*, markcolor2, mark size=1, mark options={solid}, forget plot]
table {%
-1.76705961591877 3.84971410662474
};
\addplot [semithick, mark=square*, markcolor2, mark size=1, mark options={solid}, forget plot]
table {%
-0.836579717137226 2.52650771343239
};
\addplot [semithick, mark=square*, markcolor2, mark size=1, mark options={solid}, forget plot]
table {%
-8.92370796889198 -5.73295268893707
};
\addplot [semithick, mark=square*, markcolor2, mark size=1, mark options={solid}, forget plot]
table {%
-3.28342135670187 0.57203765276991
};
\addplot [semithick, mark=square*, markcolor2, mark size=1, mark options={solid}, forget plot]
table {%
3.59805244505901 9.36971106514372
};
\addplot [semithick, mark=square*, markcolor2, mark size=1, mark options={solid}, forget plot]
table {%
3.45663147493086 3.60642637831652
};
\addplot [semithick, mark=square*, markcolor2, mark size=1, mark options={solid}, forget plot]
table {%
5.12061860827654 3.00561448405937
};
\addplot [semithick, mark=square*, markcolor2, mark size=1, mark options={solid}, forget plot]
table {%
-1.05935535152749 2.9693770238203
};
\addplot [semithick, mark=square*, markcolor2, mark size=1, mark options={solid}, forget plot]
table {%
-0.400595693890529 -8.56889560169258
};
\addplot [semithick, mark=square*, markcolor2, mark size=1, mark options={solid}, forget plot]
table {%
0.519580634969263 9.08029211797773
};
\addplot [semithick, mark=square*, markcolor2, mark size=1, mark options={solid}, forget plot]
table {%
4.73595131064313 0.720656530778581
};
\addplot [semithick, mark=square*, markcolor2, mark size=1, mark options={solid}, forget plot]
table {%
1.1040781772957 1.18293416250506
};
\addplot [semithick, mark=square*, markcolor2, mark size=1, mark options={solid}, forget plot]
table {%
5.01108470223204 2.60928679752431
};
\addplot [semithick, mark=square*, markcolor2, mark size=1, mark options={solid}, forget plot]
table {%
-1.81944263878665 4.69000455044485
};
\addplot [semithick, mark=square*, markcolor2, mark size=1, mark options={solid}, forget plot]
table {%
6.35135905179381 4.46433834551098
};
\addplot [semithick, mark=square*, markcolor2, mark size=1, mark options={solid}, forget plot]
table {%
2.97365091081021 6.23449198877002
};
\addplot [semithick, mark=square*, markcolor2, mark size=1, mark options={solid}, forget plot]
table {%
5.25901647442393 2.31252227557231
};
\addplot [semithick, mark=square*, markcolor2, mark size=1, mark options={solid}, forget plot]
table {%
-5.85951557603237 6.14285227939429
};
\addplot [semithick, mark=square*, markcolor2, mark size=1, mark options={solid}, forget plot]
table {%
1.19470511985711 -3.85298634905214
};
\addplot [semithick, mark=square*, markcolor2, mark size=1, mark options={solid}, forget plot]
table {%
3.01280176684569 -3.32688544548977
};
\addplot [semithick, mark=square*, markcolor2, mark size=1, mark options={solid}, forget plot]
table {%
-3.41501591355984 4.00271728948772
};
\addplot [semithick, mark=square*, markcolor2, mark size=1, mark options={solid}, forget plot]
table {%
-4.99718222586574 4.02762521003818
};
\addplot [semithick, mark=pentagon*, markcolor3,  mark size=1, mark options={solid}, only marks]
table {%
5.86409042806465 -4.79931766154839
};
\addlegendentry{$\mc C_{3}$}
\addplot [semithick, mark=pentagon*, markcolor3,  mark size=1, mark options={solid}, forget plot]
table {%
-5.17798635653915 7.16895411192711
};
\addplot [semithick, mark=pentagon*, markcolor3,  mark size=1, mark options={solid}, forget plot]
table {%
4.74583597906079 -0.772800972257602
};
\addplot [semithick, mark=pentagon*, markcolor3,  mark size=1, mark options={solid}, forget plot]
table {%
-10.6806499846797 2.95682467634206
};
\addplot [semithick, mark=pentagon*, markcolor3,  mark size=1, mark options={solid}, forget plot]
table {%
-1.15524919695037 4.21097656155845
};
\addplot [semithick, mark=pentagon*, markcolor3,  mark size=1, mark options={solid}, forget plot]
table {%
4.28875735841458 -4.52258973394335
};
\addplot [semithick, mark=pentagon*, markcolor3,  mark size=1, mark options={solid}, forget plot]
table {%
-4.00543869529837 -5.88389838545296
};
\addplot [semithick, mark=pentagon*, markcolor3,  mark size=1, mark options={solid}, forget plot]
table {%
-8.29754931579562 -0.370880999586282
};
\addplot [semithick, mark=pentagon*, markcolor3,  mark size=1, mark options={solid}, forget plot]
table {%
6.12289565502555 0.785982549416965
};
\addplot [semithick, mark=pentagon*, markcolor3,  mark size=1, mark options={solid}, forget plot]
table {%
-2.57375458962797 0.0181608138234862
};
\addplot [semithick, mark=pentagon*, markcolor3,  mark size=1, mark options={solid}, forget plot]
table {%
0.243263287369567 4.07435949310627
};
\addplot [semithick, mark=pentagon*, markcolor3,  mark size=1, mark options={solid}, forget plot]
table {%
3.0238882658163 1.4845198795681
};
\addplot [semithick, mark=pentagon*, markcolor3,  mark size=1, mark options={solid}, forget plot]
table {%
-2.58917122153271 -1.65169618664406
};
\addplot [semithick, mark=pentagon*, markcolor3,  mark size=1, mark options={solid}, forget plot]
table {%
3.05658849645186 -5.52381928463494
};
\addplot [semithick, mark=pentagon*, markcolor3,  mark size=1, mark options={solid}, forget plot]
table {%
1.34671376638011 3.32166803713762
};
\addplot [semithick, mark=pentagon*, markcolor3,  mark size=1, mark options={solid}, forget plot]
table {%
-2.62511307949792 -5.26795531148264
};
\addplot [semithick, mark=pentagon*, markcolor3,  mark size=1, mark options={solid}, forget plot]
table {%
0.868456507629695 -6.49046888281279
};
\addplot [semithick, mark=pentagon*, markcolor3,  mark size=1, mark options={solid}, forget plot]
table {%
-4.40574981234998 -7.29681834855236
};
\addplot [semithick, mark=pentagon*, markcolor3,  mark size=1, mark options={solid}, forget plot]
table {%
-0.186111692515278 -3.52232721277447
};
\addplot [semithick, mark=pentagon*, markcolor3,  mark size=1, mark options={solid}, forget plot]
table {%
0.523131908056405 -1.28114116228664
};
\addplot [semithick, mark=pentagon*, markcolor3,  mark size=1, mark options={solid}, forget plot]
table {%
8.8781246630023 -5.83918603213863
};
\addplot [semithick, mark=pentagon*, markcolor3,  mark size=1, mark options={solid}, forget plot]
table {%
3.71736879029653 3.63351059016848
};
\addplot [semithick, mark=pentagon*, markcolor3,  mark size=1, mark options={solid}, forget plot]
table {%
0.663056503033056 -4.11136635024403
};
\addplot [semithick, mark=pentagon*, markcolor3,  mark size=1, mark options={solid}, forget plot]
table {%
-2.89244752209947 1.03892728964789
};
\addplot [semithick, mark=pentagon*, markcolor3,  mark size=1, mark options={solid}, forget plot]
table {%
-6.85185526780689 3.90475760369411
};
\addplot [semithick, mark=pentagon*, markcolor3,  mark size=1, mark options={solid}, forget plot]
table {%
0.00259710853988992 0.651963716052784
};
\addplot [semithick, mark=pentagon*, markcolor3,  mark size=1, mark options={solid}, forget plot]
table {%
-3.4911738860573 -3.63272070242322
};
\addplot [semithick, mark=pentagon*, markcolor3,  mark size=1, mark options={solid}, forget plot]
table {%
-2.46325673327157 -2.30798161618553
};
\addplot [semithick, mark=pentagon*, markcolor3,  mark size=1, mark options={solid}, forget plot]
table {%
4.4419499837081 1.9886778806867
};
\addplot [semithick, mark=pentagon*, markcolor3,  mark size=1, mark options={solid}, forget plot]
table {%
-8.94842902132484 4.7135442120778
};
\addplot [semithick, mark=pentagon*, markcolor3,  mark size=1, mark options={solid}, forget plot]
table {%
4.23380781705056 5.44973419845131
};
\addplot [semithick, mark=pentagon*, markcolor3,  mark size=1, mark options={solid}, forget plot]
table {%
6.20154588427867 8.51380098929169
};
\addplot [semithick, mark=pentagon*, markcolor3,  mark size=1, mark options={solid}, forget plot]
table {%
8.00717800315883 -0.700716723450428
};
\addplot [semithick, mark=pentagon*, markcolor3,  mark size=1, mark options={solid}, forget plot]
table {%
0.862605110189044 2.51762829556081
};
\addplot [semithick, mark=pentagon*, markcolor3,  mark size=1, mark options={solid}, forget plot]
table {%
-1.25464254984365 1.07788486186734
};
\addplot [semithick, mark=pentagon*, markcolor3,  mark size=1, mark options={solid}, forget plot]
table {%
-0.995957268283173 8.12747230931386
};
\addplot [semithick, mark=pentagon*, markcolor3,  mark size=1, mark options={solid}, forget plot]
table {%
-6.77925541770909 -1.22901506852504
};
\addplot [semithick, mark=pentagon*, markcolor3,  mark size=1, mark options={solid}, forget plot]
table {%
1.03469684646747 0.0412464608593054
};
\addplot [semithick, mark=pentagon*, markcolor3,  mark size=1, mark options={solid}, forget plot]
table {%
-1.35338079317544 4.54155991732044
};
\addplot [semithick, mark=pentagon*, markcolor3,  mark size=1, mark options={solid}, forget plot]
table {%
4.26254731074913 -2.30811472058957
};
\addplot [semithick, mark=pentagon*, markcolor3,  mark size=1, mark options={solid}, forget plot]
table {%
-4.00060348775611 2.27913261023406
};
\addplot [semithick, mark=pentagon*, markcolor3,  mark size=1, mark options={solid}, forget plot]
table {%
-0.156083135417289 11.8938908230256
};
\addplot [semithick, mark=pentagon*, markcolor3,  mark size=1, mark options={solid}, forget plot]
table {%
-7.24922068614272 0.463028826203918
};
\addplot [semithick, mark=pentagon*, markcolor3,  mark size=1, mark options={solid}, forget plot]
table {%
13.4308634384968 -5.47833480053558
};
\addplot [semithick, mark=pentagon*, markcolor3,  mark size=1, mark options={solid}, forget plot]
table {%
1.59482094748196 -2.79424674557733
};
\addplot [semithick, color0, mark=*, mark size=3, mark options={solid}, only marks]
table {%
-48.8530559598882 -10.7484840341848
};
\addlegendentry{$\mc M_{1}$}
\addplot [semithick, color1, mark=*, mark size=3, mark options={solid}, only marks]
table {%
-2.76815037378918 5.70483416373385
};
\addlegendentry{$\mc M_{2}$}
\addplot [semithick, color2, mark=*, mark size=3, mark options={solid}, only marks]
table {%
31.5270004128829 -1.76767259053627
};
\addlegendentry{$\mc M_{3}$}
\addplot [semithick, color3, mark=*, mark size=3, mark options={solid},only marks]
table {%
45.9238361585426 20.8557665828358
};
\addlegendentry{$\mc M_{4}$}
\end{axis}

\end{tikzpicture}

%% file: figures/JgSigma.tikz
% This file was created with tikzplotlib v0.9.15.
\usepgfplotslibrary{groupplots}

\begin{tikzpicture}

\definecolor{color0}{rgb}{0.12156862745098,0.466666666666667,0.705882352941177}
\definecolor{color1}{rgb}{1,0.498039215686275,0.0549019607843137}
\definecolor{color2}{rgb}{0.172549019607843,0.627450980392157,0.172549019607843}
\definecolor{color3}{rgb}{0.83921568627451,0.152941176470588,0.156862745098039}
\definecolor{color4}{rgb}{0,0,0}

\begin{groupplot}[group style={group size=1 by 2}]
\nextgroupplot[
scaled x ticks=manual:{}{\pgfmathparse{#1}},
tick align=outside,
tick pos=left,
x grid style={white!69.0196078431373!black},
xmajorgrids,
xmin=1, xmax=2000,
xtick style={color=black},
xticklabels={},
y grid style={white!69.0196078431373!black},
ylabel={\(\displaystyle J_{G}(\sigma(x))\)},
ymajorgrids,
ymin=0, ymax=2331.74999993295,
ytick style={color=black},
height=3cm,
width=8cm,
]
\addplot [semithick, color4]
table {%
0 3214
13 2831.615234375
26 2499.9013671875
39 2212.125
52 1962.44775390625
65 1745.80871582031
78 1557.81994628906
91 1394.67700195312
104 1253.08154296875
117 1130.173828125
130 1023.47442626953
143 930.83349609375
156 850.386962890625
169 780.518615722656
182 719.827087402344
195 667.096984863281
208 621.274658203125
221 581.4462890625
224 575.31591796875
234 561.696899414062
255 535.819702148438
276 512.585998535156
298 490.786285400391
320 471.31201171875
343 453.169311523438
367 436.384521484375
392 420.961395263672
418 406.883605957031
446 393.677764892578
475 381.854736328125
506 371.011810302734
540 360.936798095703
576 352.013427734375
616 343.843200683594
660 336.585479736328
709 330.204071044922
765 324.615997314453
830 319.840637207031
908 315.845306396484
1004 312.669006347656
1129 310.285217285156
1308 308.667999267578
1616 307.789093017578
2498 307.545379638672
2999 307.542053222656
};

\nextgroupplot[
legend cell align={left},
legend style={fill opacity=0.8, draw opacity=1, text opacity=1, draw=white!80!black,at={(1,1.01)}},
tick align=outside,
tick pos=left,
x grid style={white!69.0196078431373!black},
xlabel={Iteration [k]},
xmajorgrids,
xmin=1, xmax=3000,
xtick style={color=black},
y grid style={white!69.0196078431373!black},
ylabel={$\sigma(x)$},
ymajorgrids,
ymin=12.9064516332694, ymax=60,
ytick style={color=black},
height=3.5cm,
width=8cm,
]
\addplot [semithick, color0]
table {%
1 31.9910945892334
48 31.592565536499
97 31.2161312103271
149 30.8560485839844
203 30.5207271575928
261 30.1997318267822
322 29.9010066986084
387 29.6213607788086
457 29.3592166900635
532 29.1172294616699
614 28.892068862915
703 28.6870574951172
801 28.5007591247559
909 28.3345909118652
1031 28.1863441467285
1169 28.0579090118408
1329 27.9481830596924
1519 27.8571071624756
1753 27.7843475341797
2059 27.7293758392334
2492 27.6923999786377
2999 27.6757354736328
};
\addlegendentry{$\sigma_{1}(x)$}
\addplot [semithick, color0, dashed, forget plot]
table {%
1 35
2999 35
};
\addplot [semithick, color1]
table {%
1 37.8965644836426
9 37.0892333984375
18 36.222412109375
27 35.3974151611328
36 34.6122283935547
45 33.8649291992188
54 33.1536903381348
64 32.4036026000977
74 31.6936187744141
84 31.0215969085693
94 30.3855075836182
105 29.7250194549561
116 29.1032733917236
127 28.5179958343506
139 27.9185943603516
151 27.357442855835
164 26.7898635864258
177 26.261417388916
191 25.7329921722412
205 25.2436962127686
220 24.7595806121826
223 24.6806564331055
230 24.5741348266602
262 24.1385250091553
296 23.7138271331787
332 23.3032913208008
369 22.9195919036865
408 22.5532341003418
449 22.2061595916748
493 21.8726806640625
539 21.5625705718994
588 21.2706356048584
641 20.9940986633301
697 20.740535736084
758 20.5033130645752
824 20.2857151031494
896 20.0873203277588
976 19.9063930511475
1064 19.7464027404785
1164 19.6039714813232
1278 19.4809722900391
1412 19.3761367797852
1573 19.2901496887207
1775 19.2226085662842
2044 19.1734561920166
2446 19.1420154571533
2999 19.1290321350098
};
\addlegendentry{$\sigma_{2}(x)$}
\addplot [semithick, color1, dashed, forget plot]
table {%
1 15
2999 15
};
\addplot [semithick, color2]
table {%
1 38.1034355163574
9 38.9107666015625
18 39.777587890625
27 40.6025848388672
36 41.3877716064453
45 42.1350708007812
54 42.8463096618652
64 43.5963973999023
74 44.3063812255859
84 44.978401184082
94 45.6144943237305
105 46.2749786376953
116 46.8967247009277
127 47.4820022583008
139 48.0814056396484
151 48.6425590515137
164 49.2101364135742
177 49.738582611084
191 50.2670059204102
205 50.7563056945801
220 51.2404174804688
223 51.3193435668945
230 51.4258651733398
262 51.8614730834961
296 52.2861747741699
332 52.6967086791992
369 53.0804061889648
408 53.4467658996582
449 53.7938385009766
493 54.1273193359375
539 54.4374313354492
588 54.7293663024902
641 55.0059013366699
697 55.259464263916
758 55.4966888427734
824 55.714282989502
896 55.9126777648926
975 56.0915718078613
1063 56.2519760131836
1162 56.3935279846191
1276 56.5171699523926
1409 56.6218948364258
1570 56.708553314209
1768 56.7755966186523
2032 56.8250045776367
2427 56.857120513916
2999 56.8709678649902
};
\addlegendentry{$\sigma_{3}(x)$}
\addplot [semithick, color2, dashed, forget plot]
table {%
1 50
2999 50
};
\addplot [semithick, color3]
table {%
1 32.008903503418
47 32.3993492126465
96 32.7765579223633
147 33.130802154541
201 33.4675064086914
258 33.7845878601074
319 34.0851593017578
384 34.3665390014648
454 34.6303062438965
529 34.8737907409668
610 35.0978050231934
698 35.3023948669434
795 35.488883972168
902 35.6557197570801
1023 35.8050346374512
1161 35.9356002807617
1321 36.0471458435059
1511 36.1397323608398
1741 36.2127113342285
2044 36.2686614990234
2492 36.3076019287109
2999 36.3242645263672
};
\addlegendentry{$\sigma_{4}(x)$}
\addplot [semithick, color3, dashed, forget plot]
table {%
1 40
2999 40
};
\end{groupplot}
\end{tikzpicture}

%% file: figures/robustness.tikz
\begin{tikzpicture}
\definecolor{color0}{rgb}{0.12156862745098,0.466666666666667,0.705882352941177}

\begin{axis}[
legend cell align={left},
legend style={fill opacity=0.8, draw opacity=1, text opacity=1, at={(1,1)}, draw=white!80!black},
tick align=outside,
tick pos=left,
x grid style={white!69.0196078431373!black},
xlabel={$\alpha$},
xmajorgrids,
xmin=0, xmax=1,
xtick style={color=black},
y grid style={white!69.0196078431373!black},
ylabel={$J_G\left(\sigma\left(x^{*}\right)\right)$},
ymajorgrids,
ymin=250, ymax=1000,
ytick style={color=black},
width=6cm,
height=2.5cm,
scale only axis
]
\path [fill=color0, fill opacity=0.2]
(axis cs:0,307.54204347855)
--(axis cs:0,307.541666666667)
--(axis cs:0.02,307.541750769875)
--(axis cs:0.04,307.543411683042)
--(axis cs:0.06,307.54395231032)
--(axis cs:0.08,307.547468514855)
--(axis cs:0.1,307.543533236821)
--(axis cs:0.12,307.543084246356)
--(axis cs:0.14,307.547447165175)
--(axis cs:0.16,307.544566674233)
--(axis cs:0.18,307.554477706717)
--(axis cs:0.2,307.600517318439)
--(axis cs:0.22,307.557457726101)
--(axis cs:0.24,307.55540258996)
--(axis cs:0.26,307.568064797637)
--(axis cs:0.28,307.57477290172)
--(axis cs:0.3,307.572688668571)
--(axis cs:0.32,307.730890900641)
--(axis cs:0.34,307.675751872636)
--(axis cs:0.36,307.595927844853)
--(axis cs:0.38,307.630548881414)
--(axis cs:0.4,307.924294150293)
--(axis cs:0.42,308.211219538683)
--(axis cs:0.44,308.205259102053)
--(axis cs:0.46,307.815622632583)
--(axis cs:0.48,307.669031522638)
--(axis cs:0.5,307.612046118769)
--(axis cs:0.52,307.638408524728)
--(axis cs:0.54,307.788871116432)
--(axis cs:0.56,307.63070749689)
--(axis cs:0.58,307.799104098857)
--(axis cs:0.6,307.777063740226)
--(axis cs:0.62,307.680276000258)
--(axis cs:0.64,307.562871886978)
--(axis cs:0.66,307.934001407845)
--(axis cs:0.68,308.508289023673)
--(axis cs:0.7,307.885088807658)
--(axis cs:0.72,307.915344716367)
--(axis cs:0.74,307.963697144023)
--(axis cs:0.76,307.558789150328)
--(axis cs:0.78,307.611221247201)
--(axis cs:0.8,308.290610379108)
--(axis cs:0.82,308.326406627607)
--(axis cs:0.84,308.369438087108)
--(axis cs:0.86,308.311413501)
--(axis cs:0.88,308.475169254503)
--(axis cs:0.9,308.37272225264)
--(axis cs:0.92,307.775898325359)
--(axis cs:0.94,307.814941514825)
--(axis cs:0.96,310.223729242653)
--(axis cs:0.98,309.098300658862)
--(axis cs:1,307.846433963435)
--(axis cs:1,1256.17621916934)
--(axis cs:1,1256.17621916934)
--(axis cs:0.98,1095.03750080163)
--(axis cs:0.96,829.147470496178)
--(axis cs:0.94,1136.4063634847)
--(axis cs:0.92,856.282063913467)
--(axis cs:0.9,696.835027293988)
--(axis cs:0.88,766.667994458673)
--(axis cs:0.86,722.267608361667)
--(axis cs:0.84,759.458071989746)
--(axis cs:0.82,695.277137008515)
--(axis cs:0.8,885.009561628473)
--(axis cs:0.78,701.173419072453)
--(axis cs:0.76,663.849461865276)
--(axis cs:0.74,626.583125058676)
--(axis cs:0.72,744.505044791555)
--(axis cs:0.7,664.709407326805)
--(axis cs:0.68,582.050462341987)
--(axis cs:0.66,552.512843673624)
--(axis cs:0.64,683.573132147604)
--(axis cs:0.62,515.281777657732)
--(axis cs:0.6,472.826883304557)
--(axis cs:0.58,573.194961673031)
--(axis cs:0.56,474.801898120511)
--(axis cs:0.54,468.30158381428)
--(axis cs:0.52,400.455884287157)
--(axis cs:0.5,498.49563818241)
--(axis cs:0.48,430.558271004135)
--(axis cs:0.46,420.899213114878)
--(axis cs:0.44,412.561739485018)
--(axis cs:0.42,401.347029609526)
--(axis cs:0.4,434.385340424769)
--(axis cs:0.38,433.643389597307)
--(axis cs:0.36,375.729641739875)
--(axis cs:0.34,390.028096238393)
--(axis cs:0.32,364.848963818671)
--(axis cs:0.3,349.312294076597)
--(axis cs:0.28,356.761485450364)
--(axis cs:0.26,343.590699547504)
--(axis cs:0.24,348.443351807584)
--(axis cs:0.22,334.959395974837)
--(axis cs:0.2,324.566374338771)
--(axis cs:0.18,322.173871262415)
--(axis cs:0.16,330.182790254336)
--(axis cs:0.14,320.261330770787)
--(axis cs:0.12,315.43882025878)
--(axis cs:0.1,312.429453054093)
--(axis cs:0.08,311.911622527332)
--(axis cs:0.06,309.025943529217)
--(axis cs:0.04,308.253974242082)
--(axis cs:0.02,307.65497434279)
--(axis cs:0,307.54204347855)
--cycle;

\addplot [semithick, color0]
table {%
0 307.541656494141
0.0199999809265137 307.56494140625
0.0399999618530273 307.667205810547
0.059999942779541 307.898284912109
0.0800000429153442 308.08935546875
0.100000023841858 308.391815185547
0.120000004768372 308.966888427734
0.139999985694885 309.360046386719
0.159999966621399 309.973571777344
0.180000066757202 310.619903564453
0.200000047683716 310.999481201172
0.220000028610229 312.714996337891
0.240000009536743 313.656280517578
0.259999990463257 314.440124511719
0.279999971389771 315.692626953125
0.299999952316284 316.103240966797
0.319999933242798 316.993957519531
0.340000033378601 320.458282470703
0.360000014305115 321.384582519531
0.379999995231628 323.635864257812
0.399999976158142 323.542449951172
0.419999957084656 327.828521728516
0.440000057220459 327.246124267578
0.460000038146973 333.309692382812
0.480000019073486 332.593231201172
0.5 333.703704833984
0.519999980926514 332.575347900391
0.539999961853027 344.728149414062
0.559999942779541 341.561248779297
0.579999923706055 346.881225585938
0.600000023841858 348.194793701172
0.620000004768372 347.806640625
0.639999985694885 353.116333007812
0.660000085830688 355.047119140625
0.680000066757202 357.635375976562
0.700000047683716 359.6513671875
0.720000028610229 368.817993164062
0.740000009536743 368.152984619141
0.759999990463257 372.766784667969
0.779999971389771 376.443115234375
0.799999952316284 397.569396972656
0.819999933242798 392.659698486328
0.839999914169312 386.850433349609
0.860000014305115 397.187225341797
0.879999995231628 394.451019287109
0.899999976158142 389.699798583984
0.920000076293945 401.026733398438
0.940000057220459 435.972015380859
0.960000038146973 436.870239257812
0.980000019073486 417.500335693359
1 449.406402587891
};
\addlegendentry{Mean}
\end{axis}

\end{tikzpicture}